\documentclass[12pt]{article}
\usepackage{a4wide}

\usepackage{plain}
\usepackage{latexsym}
\usepackage{amsthm,bbm}
\usepackage[mathscr]{eucal}
\usepackage{amsmath,amssymb,mathrsfs}
\theoremstyle{plain}
\newtheorem{theorem}{Theorem}[section]
\newtheorem{lemma}[theorem]{Lemma}
\newtheorem{corollary}[theorem]{Corollary}

\begin{document}

\providecommand{\tr}{\ensuremath{\mathrm{tr}}}
\providecommand{\supp}{\ensuremath{\mathrm{supp}\ }}
\providecommand{\ad}{\ensuremath{\mathrm{ad} }}
\providecommand{\Ks}{\ensuremath{\mathcal{K}}}
\providecommand{\Hs}{\ensuremath{\mathcal{H}}}
\providecommand{\B}{\ensuremath{\mathcal{B}}}
\providecommand{\Ss}{\ensuremath{\mathcal{S}}}
\providecommand{\C}{\ensuremath{\mathbb{C}}}
\providecommand{\R}{\ensuremath{\mathbb{R}}}
\providecommand{\Z}{\ensuremath{\mathbb{Z}}}
\providecommand{\N}{\ensuremath{\mathbb{N}}}
\providecommand{\1}{\ensuremath{\mathbbm{1}}}
\providecommand{\T}{\ensuremath{\boldsymbol{T}}}
\providecommand{\A}{\ensuremath{\mathfrak{A}}}
\providecommand{\W}{\ensuremath{\mathfrak{W}}}
\providecommand{\xv}{\ensuremath{\mathbf{x}}}
\providecommand{\kv}{\ensuremath{\mathbf{k}}}
\providecommand{\pv}{\ensuremath{\mathbf{p}}}
\providecommand{\0}{\ensuremath{\mathbf{0}}}
\providecommand{\dif}{\ensuremath{\mathrm{d}}}
\providecommand{\D}{\ensuremath{\mathrm{d}}}
\providecommand{\tensor}{\ensuremath{\otimes}}
\providecommand{\com}[1]{{\bf \ $\star $ #1 $\star$\ }}
\providecommand{\abs}[1]{\ensuremath{\left| #1 \right|}}
\providecommand{\betrag}[1]{\ensuremath{\left| #1 \right|}}
\providecommand{\norm}[1]{\ensuremath{\left\| #1 \right\|}}
\providecommand{\wick}[1]{\ensuremath{\text{\bf :} #1 \text{\bf :}}}
\providecommand{\bra}[1]{\ensuremath{\left\langle {#1} \right|}}
\providecommand{\ket}[1]{\ensuremath{\left|{#1} \right\rangle}}

\noindent
\begin{center}
{ \Large \bf Quantum Energy Inequalities for the \\Non-Minimally Coupled Scalar Field}
\\[35pt]
{\large \sc Christopher J.\ Fewster${}^{\dagger}$\ 
 {\rm and}\  Lutz W. Osterbrink${}^\ast$}\\[20pt]  
                 Department of Mathematics,\\
                 University of York,\\
                 Heslington,\\
                 York YO10 5DD, United Kingdom\\[4pt]
                 ${}^{\dagger}$ e-mail: cjf3$@$york.ac.uk \\
                 ${}^{\ast}$ e-mail: lwo500$@$york.ac.uk \\[18pt]
                October 23, 2007
\end{center}
${}$\\[18pt]
{\small {\bf Abstract.} 
In this paper we discuss local averages of the energy density for the non-minimally
coupled scalar quantum field, extending a previous investigation of the classical field. 
By an explicit example, we show that such averages are unbounded from below on
the class of Hadamard states. This contrasts with the minimally coupled field, which obeys
a state-independent lower bound known as a Quantum Energy Inequality (QEI). 
Nonetheless, we derive a generalised QEI for the non-minimally coupled
scalar field, in which the lower bound is
permitted to be state-dependent. This result applies to general globally hyperbolic curved
spacetimes for coupling constants in the range $0<\xi\leq 1/4$. We analyse the
state-dependence of our QEI in four-dimensional Minkowski space and show
that it is a non-trivial restriction on the averaged energy density in the
sense that the lower bound is of lower order, in energetic terms, than the
averaged energy density itself.
}

\vspace{0.3truecm}

${}$

 
\section{Introduction}
For more than 30 years it has been known that the stress-energy tensor
of the classical scalar field, obtained from the Lagrangean 
\begin{equation}\label{eqn_lagrangean}
L[\phi]=\frac{1}{2}(\nabla \phi)^2-\frac{1}{2}(m^2+\xi R)\phi^2 ,
\end{equation}
does not satisfy the weak energy condition
(WEC) at non-minimal coupling, i.e., $\xi\not=0$
(see~\cite{B74} and~\cite{FR01} for a simple example in Minkowski space).
Naturally, this raises the question of whether there are any restrictions on the extent of WEC violation and whether this field could support exotic phenomena (e.g., violations of the second law of thermodynamics) that depend on macroscopic spacetime regions of negative energy density.

We showed in~\cite{FO06} that under certain conditions one can find lower bounds for local averages of the energy density of the classical non-minimally coupled scalar field. 
This paper extends our analysis to the case of the quantised field.

As is well known, quantum field theories obeying the Wightman axioms necessarily violate the WEC~\cite{EGJ65} and in this respect, the non-minimally coupled scalar field resembles the situation at minimal coupling.
However, as we show, there are differences. 
For example, consider  the case where $\xi>0$. Given any bounded subset~$\mathcal{O}$ of Minkowski
space and an arbitrary constant $\rho_0>0$ we will construct a Hadamard
state $\Psi$ 
in which the expected energy density $\langle\rho\rangle_\Psi$ is less than $-\rho_0$
throughout~$\mathcal{O}$.
Therefore, non-trivial local averages of the energy
density, $\langle\rho(\mathfrak{f})\rangle_\Psi$, where $\mathfrak{f}$
is a non-negative test function, are
unbounded from below on the class of Hadamard states.

By contrast, expectation values of the averaged energy density of the
{\em minimally} coupled scalar field are bounded from below on the class of Hadamard states~\cite{F00}.
The latter bound, known as a quantum energy inequality~(QEI), can be written in the form
\begin{equation}
\langle\rho (\mathfrak{f})\rangle_\Psi\geq
-\widetilde{\mathfrak{Q}}(\mathfrak{f})\qquad
\end{equation}
for all Hadamard states $\Psi$, where~$\widetilde{\mathfrak{Q}}(\mathfrak{f})$
is a constant (see~\cite{F03,R04,Ford06} for reviews
and further references concerning such QEIs). 
Clearly, the non-minimally coupled field cannot satisfy a QEI of this type.
However, we will show that it obeys a generalised QEI of the form
\begin{equation}
\langle\rho (\mathfrak{f})\rangle_\Psi \geq -\langle\mathfrak{Q}(\mathfrak{f})\rangle_\Psi
\end{equation}
for all Hadamard states $\Psi$, where~$\mathfrak{Q}(\mathfrak{f})$ is now allowed to be an unbounded
operator, which turns out to involve the Wick square of the field. 
This bound will be proved for averaging along time-like geodesics in general globally hyperbolic spacetimes. Precise statements and the proof are given in
section~\ref{mainresult}. State-dependent QEIs (and related results)
have recently been studied in an abstract context by one of
us~\cite{F06}, in which they are naturally suggested by the mathematical
framework. This paper complements
that work by giving a concrete example of a quantum field which admits a
state-dependent bound but cannot admit a state-independent one.

The state-dependent nature of the lower bound raises an important question. 
It is clear that setting~$\mathfrak{Q}(\mathfrak{f})=-\rho (\mathfrak{f})$ would provide a rather trivial inequality of the above type. Are our bounds similarly trivial?
In section~\ref{section_nontriviality} we will analyse this question in two ways. The first is based on a proposal in~\cite{F06},
in which a QEI would be declared trivial if there exist constants~$c$ and~$c'$ such that
\begin{equation}\label{eqn_nontriviality}
\abs{\langle \rho (\mathfrak{f})\rangle_\Psi} \leq c+ c' \abs{\langle\mathfrak{Q}(\mathfrak{f})\rangle_\Psi}
\end{equation}
for all Hadamard~$\Psi$. We will show that our bound is non-trivial in
this sense by considering finite temperature states in Minkowski space.
The second way uses so-called $H$-bounds: we show that
$\mathfrak{Q}(\mathfrak{f})$ can be bounded by any power of the
Hamiltonian greater than $2$, while $\rho (\mathfrak{f})$ cannot be
bounded by powers less than $3$.\footnote{The power of $3$ emerges by considering a
particular family of states, and it may be that $\rho (\mathfrak{f})$
can only be bounded by powers of at least $4$, as would be natural on
dimensional grounds. See section~\ref{section_energetics} for more discussion.} Thus,
although the lower bound is state-dependent, it is more stringent in
energetic terms than any upper bound on the averaged energy density: one
may say that the $\mathfrak{Q}(\mathfrak{f})$ is of lower order than
$\rho(\mathfrak{f})$. In particular, states that exhibit large negative energy densities over extended spacetime regions necessarily have large positive overall energy. 
Further comments on the significance of our results are given in section~\ref{conclusion}.

\section{The non-minimally coupled field}\label{the_NMCQF}
We begin by recalling the definition of the non-minimally coupled scalar field, its quantisation and the
construction of the stress-energy tensor.
This will serve to fix our conventions. 

The classical Lagrangean describing the field on a $n$-dimensional spacetime\footnote{Our
sign conventions are those of Birrell and Davies~\cite{BD82}, 
i.e., the $[-,-,-]$ convention in the classification scheme of Misner, Thorne and Wheeler~\cite{MTW73}.}
$\mathbf{M}=(M,g)$ is given by \eqref{eqn_lagrangean},
where $m,\xi$ are real constants and $R$ is the Ricci scalar with respect to the metric~$g$.
The constant $\xi$ is called the coupling constant. If~$\xi=0$, the field is said to be minimally
coupled. For~$\xi=(n-2)/(4n-4)$ and~$\xi=1/4$, one speaks of conformal and super-symmetric coupling,
respectively. In this paper we will focus on values $\xi\in [0,1/4]$, which clearly contains all the special values just mentioned. 
The Lagrangean~\eqref{eqn_lagrangean} leads to the wave equation 
\begin{equation}\label{eqn_waveeqn}
P_\xi\phi=0,
\end{equation}
where $P_\xi:= \square_g+\left(m^2 +\xi R\right)$ is the Klein-Gordon operator, and
$\square_g$ is the d'Alembert\-ian with respect to the metric~$g$.
We will follow the standard convention and denote the space of compactly supported, smooth, complex-valued functions on~$M$ by~$\mathscr{D}(M)$.
Assuming that the spacetime is globally hyperbolic, there is an antisymmetric bi-distri\-bu\-tion $E_\xi
(x,y)$ which is the difference of the advanced and retarded Green
functions. 

The theory is quantised by introducing a unital $*$-algebra $\mathfrak{A}_\xi
(\mathbf{M})$, which is generated by objects $\Phi(f)$
($f\in\mathscr{D}(M)$) subject to the relations that (a) the map $f\rightarrow \Phi (f)$
is complex linear; (b) $\Phi (f)^{*}=\Phi (\overline{f})$; (c) $\Phi (P_\xi f)=
0$; (d) $[\Phi (f),\Phi (h)]=i E_\xi (f,h)\1$ for all $f,h\in\mathscr{D}(M)$. Properties (c) and (d)
enforce the field equation and the canonical commutation relations
respectively. In this framework, states are positive and normalised linear
functionals $\langle\cdot\rangle_\Psi\rightarrow\C$ on the algebra~$\mathfrak{A}_\xi (\mathbf{M})$.
In particular, we will be interested in Hadamard states: in such a state 
$\Psi$, the two-point function $\omega_2^\Psi(x,y)=\langle \Phi (x)\Phi
(y)\rangle_\Psi$ is a distribution with a prescribed
singularity structure so that the difference between the two-point functions of
any two Hadamard states is smooth.
See~\cite{W94} and references therein, for details on Hadamard states.
Some of our later results will be based on a characterisation of the
Hadamard states in terms of microlocal analysis due to
Radzikowski~\cite{R96}.

We now turn to the problem of quantising quadratic classical
expressions of the form $G^{class}(x)=[\sum_i \hat{ D}_i(\phi\otimes
\phi)]_c(x)$, where the $\hat{D}_i$ are linear differential operators on $\mathcal{C}^\infty
(M\times M)$ with smooth coefficients and $[\cdot]_c(x)$ denotes the
`coincidence limit' $[F]_c(x):=F (x,x)$, of any smooth function
$F\in\mathcal{C}^\infty (M\times M)$.
The quantised normal ordered form of $G^{class}$ in the Hadamard state $\Psi$ is then defined by
\begin{equation}\label{eqn_class_quad_obs}
\langle G^{quant}\rangle_\Psi(x)=\Big[\sum_i \hat{D}_i\wick{\omega_2^\Psi}\Big]_c (x),
\end{equation} 
where~$\wick{\omega_2^\Psi}=\omega_2^\Psi-\omega_2^0$ is the normal ordering of $\omega_2^\Psi$
with respect to a reference Hadamard state~$\omega_0$.
In Minkowski space one has the distinguished vacuum state $\Omega$ and therefore one usually
chooses  $\omega_2^0=\omega_2^\Omega$. In all quasi-free representations, normal ordering coincides with Wick normal
ordering of annihilation and creation operators.  

The classical stress-energy tensor of the non-minimally coupled scalar
field can be calculated by varying the action of the Lagrangean~\eqref{eqn_lagrangean} with respect to the metric, and takes the form 
\begin{equation}\label{eq_st0}
T^{class}_{\mu \nu}=\left(\nabla_\mu \phi\right)\left( \nabla_\nu \phi\right) +\frac{1}{2}g_{\mu\nu} \left(m^2\phi^{2}-(\nabla \phi)^2\right)+\xi \left( g_{\mu\nu}\square_{g} -\nabla_\mu \nabla_\nu-G_{\mu\nu}\right) \phi^2,
\end{equation}
where $G_{\mu\nu}$ is the Einstein tensor. 
The term proportional to the coupling constant~$\xi$ originates in the variation of the coupling term in the Lagrangean density.
In a Ricci-flat spacetime, the differential operators in~\eqref{eq_st0} proportional to $\xi$ are still present. So minimal coupling and vanishing Ricci scalar result in the same wave equation but in a different stress energy tensor.
 
To quantise~\eqref{eq_st0} in the way we introduced above, we need to bring it into the form used in~\eqref{eqn_class_quad_obs}. This can be done with the definition of the Klein-Gordon operator $P_\xi$.
One finds that
\begin{eqnarray}\label{eq_st1}
T^{class}_{\mu\nu}&=&(1-2\xi)\left(\nabla_\mu\phi\right)\left(\nabla_\nu\phi\right)+\frac{1}{2}\left(1-4\xi\right)g_{\mu\nu}\left(m^2\phi^2-(\nabla \phi)^2\right)
\nonumber\\
&&{}-2\xi \left(\phi\nabla_\mu\nabla_\nu \phi+\frac{1}{2}R_{\mu\nu}\phi^2-\frac{1}{4}\left(1-4\xi\right)g_{\mu\nu} R\phi^2 - g_{\mu\nu}\phi P_\xi\phi \right).
\end{eqnarray}
The last term in the bottom line vanishes ``on shell'', that is for $\phi$ satisfying the wave equation~\eqref{eqn_waveeqn}.

We shall often study the energy density of~\eqref{eq_st1} with respect to
freely falling observers. Assume that $\gamma $ is a time-like geodesic parameterised by proper
time,\footnote{We require $\gamma$ to be connected, but it does not have to be inextendible.}
i.e., $\dot{\gamma}^2=1$ and $\nabla_{\dot{\gamma}}\dot{\gamma}=\dot{\gamma}^\mu\nabla_\mu\dot{\gamma} =0$. Using this, together with~\eqref{eq_st1}, one can show that the classical energy density $\rho^\textrm{class}_\phi=T^\textrm{class}_{\mu\nu}\dot{\gamma}^\mu \dot{\gamma}^\nu$ on $\gamma$ is
\begin{eqnarray}\label{eq_geodavst1}
\rho^\textrm{class}_\phi&=&(1-2\xi)\left(\nabla_{\dot{\gamma}}\phi\right)^2+\frac{1}{2}\left(1-4\xi\right)\left(m^2\phi^2-(\nabla \phi)^2\right)
\nonumber\\
&&{}-2\xi \left(\phi\nabla_{\dot{\gamma}}^2\phi+\frac{1}{2}R_{\mu\nu}\dot{\gamma}^\mu \dot{\gamma}^\nu\phi^2-\frac{1}{4}\left(1-4\xi\right)R\phi^2 -\phi P_\xi\phi \right).
\end{eqnarray}
Now let $\mathcal{T}$ be an open tubular neighbourhood of $\gamma$. Take a family of smooth vector
fields~$\{v_i\}_{i=0\dots n-1}$ on~$\mathcal{T}$, 
whose restriction to $\gamma$ is a vielbein with the property that
$v_0|_\gamma=\dot{\gamma}$, so we have~$g^{\mu\nu}|_{\gamma}=v^\mu_0v^\nu_0-\sum_{i=1}^{n-1}v_i^\mu v_i^\nu$.  
We now introduce the operators
\begin{eqnarray}
\hat{\rho}_1&=&\frac{1}{2}(\nabla_{v_0} \otimes \nabla_{v_0})+\frac{1}{2}\left(1-4\xi\right)\left(m^2(\1\otimes \1)+\sum_{i=1}^{n-1}(\nabla_{v_i} \otimes \nabla_{v_i})\right),\\
\hat{\rho}_2&=&2 (\1\otimes_\mathfrak{s}\nabla^2_{v_0}),\\
\hat{\rho}_3&=&-\left(\1\otimes_\mathfrak{s} R_{\mu\nu}v_0^\mu v_0^\nu)\1\right)+\frac{1}{2}\left(1-4\xi\right)(\1\otimes_\mathfrak{s}R\1) +2(\1\otimes_\mathfrak{s}P_\xi).
\end{eqnarray}
Here $\otimes_\mathfrak{s}$ is the symmetrised tensor product, i.e., $P\otimes_\mathfrak{s}P'=\left\{(P\otimes P')+(P'\otimes P)\right\}/2$. 
Having introduced these operators, one finds that
\begin{eqnarray}
\rho^\textrm{class}_\phi&=&\left[\hat{\rho} (\phi\otimes \phi )\right]_{c},\\
\hat{\rho}&=&\hat{\rho}_1-\xi \hat{\rho}_2+\xi \hat{ \rho}_3.\label{operator_decompos}
\end{eqnarray}
Note that~$ \hat{ \rho}_3 (\phi\otimes \phi)=0$ in Ricci-flat spacetimes if $\phi$ is a solution to the wave equation~\eqref{eqn_waveeqn}. Furthermore, for minimal coupling~($\xi=0$), we have $\hat{\rho}=\hat{\rho}_1$. 

We quantise the energy density by replacing the classical point-split field $\phi\otimes\phi$
by the normal ordered two-point function $\wick{ \omega_2^{\Psi} }$ of some Hadamard state~$\Psi$.
As noted before, normal ordering is always performed with respect to some fixed reference Hadamard state $\Psi_0$. 
So the quantised energy density on $\gamma$ in the state $\Psi$ is simply given by
$\langle \rho^\textrm{quant}\rangle_\Psi=[\hat{\rho} \wick{ \omega_2^{\Psi}
}]_c$. Note that our normal ordered energy density is not the same as the
renormalised energy density $\langle \rho^\textrm{ren}\rangle_\Psi$ obtained using
the Hadamard prescription (see,
e.g.,~\cite{W78}) but they are related by $\langle
\rho^\textrm{quant}\rangle_\Psi = \langle \rho^\textrm{ren}\rangle_\Psi
- \langle \rho^\textrm{ren}\rangle_{\Psi_0}$. 

We end this section with a short summary of the non-minimally coupled scalar quantum
field in the $n$-dimensional Minkowski space~$\mathbf{M}^n_{\textrm{Mink}}=(\R^n,\eta )$. 
We define the measure $\D \mu (\kv)$ by
\begin{equation}
\D \mu (\kv)=\frac{\D^{n-1}\kv}{(2\pi)^{n-1}} \frac{1}{2\omega (\kv)},
\end{equation} 
with $\omega (\kv)=\sqrt{\kv^2+m^2}$ and use it to define the one-particle Hilbert space
by $\mathcal{H}=L^2 \left(\R^{n-1},\D\mu (\kv)\right)$.  
We will denote the norm and inner product on $\mathcal{H}$ by
$\norm{\cdot}_\mathcal{H}$ and $\langle\cdot,\cdot\rangle_{\mathcal{H}}$,
and define the bosonic  Fock-space $\mathcal{F}_s(\mathcal{H})$ in the usual way.
Thus, for each $g\in \mathcal{H}$, we have an 
annihilation operator $a(g)$ and creation operator $a^\dagger (g)$ with common
domain $D((N+\1)^{1/2})$, where $N$ is the number operator on
$\mathcal{F}_s(\mathcal{H})$, obeying the canonical commutation
relations $[a(f),a^\dagger(g)] = \langle f,g\rangle_{\mathcal{H}}\1$
(recall that $g\mapsto a(g)$ is antilinear).

For any compactly supported distribution $f\in
\mathscr{E}'(\mathbf{M}_{\textrm{Mink}}^n)$, we define
$\widetilde{f}(\kv)=\hat{f}(\omega(\kv),\kv)$, with the Fourier transformation convention
\begin{equation}
\hat{f}(k)=\int\D^nx\ e^{i kx}f(x).
\end{equation} 
If $f$ additionally satisfies the property that $\|\widetilde{f}\|_\mathcal{H}<\infty$ and $\|\widetilde{\overline{f}}\|_\mathcal{H}<\infty$, we can define \begin{equation}
\Phi (f)=a(\widetilde{\overline{f}}) + a^\dagger (\widetilde{f})
\end{equation}
as an operator 
on $D\left((N+\1)^{1/2}\right)$. If we restrict to
smooth compactly supported $f$, the operators $\Phi(f)$ restricted to
$\cap_{k=0}^\infty D(N^k)$ generate a
representation of $\mathfrak{A}_\xi(\mathbf{M}_\textrm{Mink}^n)$ (for
any $\xi$). 

Formally, we may write 
\begin{equation}\label{eqn_def_annop_munorm}
a (g)=\int \D \mu (\kv) \ \overline{g}(\kv ) a (\kv ) \textrm{ and } a^\dagger (g)=\int \D \mu (\kv) \ g(\kv ) a^\dagger (\kv ).
\end{equation}
with the~$a(\kv)$ and~$a^\dagger (\kv)$ satisfying the commutation relations
\begin{equation}
[a(\kv), a^\dagger (\kv')]=(2\pi)^{n-1}\  2\omega (\kv) \delta
(\kv-\kv')\1,
\end{equation} 
and the smeared field may be written as
\begin{equation}
\Phi(f)= \int\D^n x\, \Phi(x)f(x) ,
\end{equation}
where
\begin{equation}\label{def_n-dim_field}
\Phi (x)=\int\D  \mu (\kv) \left( a(\kv)e^{-ikx} +a^\dagger(\kv)e^{ikx} \right),
\end{equation}
with $x=(t,\xv)$ and $k=(\omega (\kv),\kv)$. Expressions of this type can be made rigorous (cf.\
section X.7 in~\cite{RS75}) and we will make use of this later on.

\section{Quantum states with negative energy density}\label{sec_negenergexample}
In this section we consider the massless quantised scalar quantum field
with $\xi>0$ in
a~$(3+1)$-dimensional Minkowski space~$\mathbf{M}^n_{\textrm{Mink}}$.
It will be shown that local averages of the energy density are unbounded from
below on the class of Hadamard states. 

We start by considering one-particle states $\Psi$, for which we have the identity
\begin{equation}\label{onestateidentity}
\wick{\omega_2^\Psi} (t,\xv,t',\xv ')=\langle\Psi |\wick{\Phi (t,\xv)\Phi (t',\xv')}\Psi\rangle=2\textrm{Re}\left( \overline{\langle \Omega |\Phi (t,\xv)\Psi\rangle} \langle \Omega |\Phi (t',\xv') \Psi\rangle\right),
\end{equation}
as can be shown by writing the fields in terms of annihilation and creation operators. 
The expression for the renormalised energy density becomes relatively simple since the differential operator $\hat{\rho}$, defined in~\eqref{operator_decompos}, acting on bi-solutions in Minkowski space is of the form
\begin{equation}
\hat{\rho}=\frac{1}{2}(\partial_t\otimes\partial_{t '}) +\frac{1}{2}(1-4\xi)\sum_{b=1}^{3}(\partial_{b}\otimes\partial_{b'}) -\xi \left((\partial_t^2\otimes\1)+(\1\otimes \partial^2_{t'} )\right) . 
\end{equation}
So, using~\eqref{onestateidentity} we have
\begin{eqnarray}\label{eqn_onepartenergdens}
\langle\Psi |\rho^\textrm{quant}\Psi\rangle(t,\xv)&=& \abs{\partial_t\langle\Omega|\Phi (t,\xv) \Psi\rangle}^2+(1-4\xi)\sum_{b=1}^{3}\abs{\partial_{b}\langle\Omega|\Phi(t,\xv) \Psi\rangle}^2 \nonumber\\
&& -4\xi \textrm{Re} \left( \overline{\langle\Omega |\Phi(t,\xv) \Psi\rangle }\times\partial^2_t\bra{\Omega}\Phi(t,\xv) \ket{\Psi} \right).
\end{eqnarray}
For $\kappa>0$, we consider (normalised) one-particle states of the form $\Psi_\kappa=a^\dagger (h_\kappa )\Omega$ 
where $h_\kappa(\kv)=4\pi\sqrt{2}(\kappa-\abs{\kv}/3)e^{-\abs{\kv}/\kappa}/\kappa^2$, for which
\begin{eqnarray}
\langle\Omega|\Phi (t,\xv)\Psi\rangle&=&\int \D \mu (\kv) \  \ e^{-i(t\abs{\kv}-\xv\kv)}  h_\kappa(\kv) .
\end{eqnarray}
Owing to
the rapid decay of $h_\kappa$, $\Psi_\kappa$ is Hadamard. It has expected energy
\begin{equation}
\langle \Psi_\kappa| H\Psi_\kappa\rangle = \frac{2\kappa}{3}.
\end{equation} 
The functions $h_\kappa$ obey $h_{\lambda\kappa}(\kv)=h_\kappa
(\kv/\lambda)/\lambda$, as a consequence of which we have the scaling relation
\begin{equation}\label{eq:scaling}
\langle\Psi_{\lambda\kappa}|\rho^\textrm{quant} \Psi_{\lambda\kappa}\rangle (x)
=
\lambda^4\langle\Psi_\kappa|\rho^\textrm{quant}\Psi_\kappa\rangle(\lambda x).
\end{equation}
Evaluating the energy density of the state $\ket{\Psi_\kappa}$ at the spatial
origin (for example) one 
finds that\footnote{The well-known identity $\int_0^\infty \D k\ k^ne^{-k}=n!$ makes most of the calculations almost trivial.}  
\begin{equation}
\langle\Psi_\kappa|\rho^\textrm{quant}\Psi_\kappa\rangle (t,\mathbf{0})
=\frac{8\kappa^4}{3(1+t^2\kappa^2)^5\pi^2}\left\{(3t^4\kappa^4+3t^2\kappa^2) -\xi(18 t^4\kappa^4-44t^2\kappa^2+2)\right\},
\end{equation}
so, in particular, 
\begin{equation}\label{eq_negenergatpoint}
\langle\Psi_\kappa|\rho^\textrm{quant}\Psi_\kappa\rangle(0,\mathbf{0})=-\xi \frac{(2\kappa)^4}{3\pi^2}.
\end{equation}
Owing to continuity of the expected energy density and
\eqref{eq_negenergatpoint} it follows that to every $\kappa>0$, there exists a constant $\tau >0$ such
\begin{equation}
\langle\Psi_\kappa|\rho^\textrm{quant}\Psi_\kappa\rangle(x)\leq -\xi
\frac{(2\kappa)^4}{6\pi^2}\qquad\textrm{for all $x\in B(\tau)$},
\end{equation}
where $B(\tau)$ is the open ball
\begin{equation}
B(\tau) = \{(t,\xv)| t^2+\abs{\xv}^2<\tau^2\}.
\end{equation}
Now fix some $\kappa$ and some appropriate $\tau$. 
As a consequence of \eqref{eq:scaling} we can find a $\kappa '>0$ to every (arbitrary) $\tau '>0$, such that
\begin{equation}
\langle\Psi_{\kappa'}|\rho^\textrm{quant}\Psi_{\kappa'}\rangle(x)\leq -\xi
\frac{(2\kappa ')^4}{6\pi^2}\qquad\textrm{for all $x\in B(\tau' )$}.
\end{equation}
In particular, we can take $\kappa '=\kappa \tau/\tau '$. We
have therefore constructed a Hadamard state with energy density less
than $-\xi\kappa^4\tau^4/(6\pi^2(\tau') ^4)$ on an arbitrary region $B(\tau')$ [by translational
invariance, the same applies to any other spacetime ball of
radius $\tau'$]. The total expected energy of this state is $2\kappa\tau/(3\tau' )$.

Note that the product of $\kappa'$ and $\tau'$ is constant.
This shows that we may arrange for large regions of negative energy
density albeit with low magnitude. We may extend the example as
follows. Suppose a constant energy density $\rho_0>0$ is given, and choose
an integer $j>6\pi^2\rho_0/(\xi(\kappa')^4)$. Then the $j$-particle
state $\Psi_{\kappa'}^{\otimes
j}=\Psi_{\kappa '}\otimes\dots\otimes \Psi_{\kappa' }$ has
energy density
\begin{equation}
\langle\Psi^{\otimes j}_{\kappa '}|\rho^\textrm{quant}\Psi^{\otimes
j}_{\kappa '}\rangle(x) = j \langle\Psi_{\kappa '}|\rho^\textrm{quant}\Psi_{\kappa' }\rangle(x)
<-\rho_0\textrm{ for all }x\in B(\tau '),
\end{equation}
and total energy
\begin{equation}
\langle \Psi_{\kappa'}^{\otimes j}| H\Psi_{\kappa'}^{\otimes j}\rangle =
\frac{2j\kappa '}{3} > \frac{2\pi^2\rho_0}{\xi (\kappa ')^3}=(\tau' )^3\frac{2\pi^2\rho_0}{\xi \kappa^3\tau^3},
\end{equation}
illustrating that the large negative energy density effects also require
large positive overall energy, at least in this example. We will see
later that this is a general phenomenon. 

Summarising, we have shown that to any bounded subset~$\mathcal{O}$ of Minkowski
space and arbitrary constant $\rho_0>0$ there is a Hadamard
state in which the expected energy density is less than $-\rho_0$
throughout~$\mathcal{O}$. In particular any smearing
$\rho^\textrm{quant}(\mathfrak{f})$ with a non-negative compactly
supported distribution $\mathfrak{f}$ is unbounded from below on the
class of Hadamard states. 

\section{Quantum energy inequalities}\label{mainresult}
In this section we are going to derive the main result that is to give a lower bound for time-like averages of the energy density.
We start with a quantum field on a curved spacetime. In a second step we specialise these results
to Minkowski space, where the vacuum state is the preferred reference state.

\subsection{Globally hyperbolic spacetime}
We keep the same assumptions as in section~\ref{the_NMCQF}; in particular~$\gamma $ is a
time-like, connected geodesic parameterised by proper time.  
Our goal is to find a lower bound for the weighted average of the quantum energy density on~$\gamma$, 
\begin{equation}\label{eqn_generalaverage}
 \langle \rho^\textrm{quant}\circ\gamma\rangle_\Psi  (\mathfrak{f})
=\int \D \tau \ \mathfrak{f}(\tau)\  \langle \rho^\textrm{quant}\rangle_\Psi  (\gamma (\tau))
\end{equation}
in the case where~$\mathfrak{f}=f^2$ for some real valued function~$f\in
\mathscr{D}(\R,\R)$. The main task is to rewrite~\eqref{eqn_generalaverage} in
such a way that the lower bound may be deduced by discarding manifestly
positive terms. 

To start, we use~\eqref{operator_decompos} to get 
\begin{eqnarray}\label{first_glob_hyp_aver}
\lefteqn{\langle \rho^\textrm{quant}\circ\gamma\rangle_\Psi  (f^2)}\nonumber\\
&=& \left( [\hat{\rho}_1\  \wick{\omega_2^\Psi}]_c
\circ\gamma\right) (f^2) -\xi  \left(  [\hat{\rho}_2\  \wick{\omega_2^\Psi}]_c\circ\gamma\right) (f^2)+ \xi\left( [\hat{ \rho}_3\  \wick{\omega_2^\Psi}]_c \circ\gamma\right) (f^2).
\end{eqnarray}
Each term on the right-hand side will be treated in turn. It will also
be useful to define $\varphi (\tau,\tau')=(\gamma (\tau),\gamma (\tau '
))$, and to write $\varphi^*F$ to denote the pull-back $\varphi^*F(\tau,\tau')
=F(\gamma (\tau),\gamma (\tau'))$ of a smooth function $F$ from $M\times
M$ to $\R\times \R$.

The first term on the right-hand side in~\eqref{first_glob_hyp_aver} may
then be rewritten, following~\cite{F00}, as
\begin{eqnarray}\label{eqn_rho1intident}
\lefteqn{\left( [\hat{\rho}_1\  \wick{\omega_2^\Psi}]_c
\circ\gamma\right) (f^2)}\nonumber\\
&=&\int \D \tau \ f^2(\tau )\ \varphi^*(\hat{\rho}_1\  \wick{\omega_2^\Psi}) (\tau,\tau)\nonumber\\
&=&\int \D \tau \D \tau ' \ \delta (\tau-\tau ')f(\tau)f(\tau') \varphi^*(\hat{\rho}_1\  \wick{\omega_2^\Psi})(\tau, \tau ')\nonumber\\
&=&\int_0^\infty \frac{\D \alpha}{\pi}\int \D \tau \D \tau ' \ e^{-i\alpha (\tau -\tau ')}f(\tau)f(\tau ') \varphi^*(\hat{\rho}_1\  \wick{\omega_2^\Psi})(\tau, \tau ')\nonumber\\
&=&\int_0^\infty \frac{\D \alpha}{\pi}  \varphi^*(\hat{\rho}_1\  \wick{\omega_2^\Psi})\left(\overline{f_\alpha} , f_\alpha \right)\nonumber\\
&=&\int_0^\infty \frac{\D \alpha}{\pi}  \varphi^*(\hat{\rho}_1\  \omega_2^\Psi)\left(\overline{f_\alpha} , f_\alpha \right)-\int_0^\infty \frac{\D \alpha}{\pi}  \varphi^*(\hat{\rho}_1\  \omega_2^0)\left(\overline{f_\alpha} , f_\alpha \right),
\end{eqnarray}
where $f_\alpha (\tau)=e^{i\alpha\tau}f (\tau)$. Here,  we have made use
of the Fourier representation of the $\delta$-function and also 
the symmetry of the normal ordered two-point function to arrange that
the $\alpha$-integral takes place over the positive half-axis. 
A decomposition of this kind will be referred to as a {\it point-splitting trick}, see~\cite{F00}.
The distributional pull-backs of the form $\varphi^*(\hat{\rho}_1\ 
\omega_2^\Psi)$ appearing in the last step were shown to exist in~\cite{F00}, using the
microlocal characterisation of Hadamard states given in~\cite{R96}. 
Moreover, if $\xi \leq 1/4$
then these distributions are positive type, i.e., $\varphi^*(\hat{\rho}_1\ 
\omega_2^\Psi)(\overline{\zeta},\zeta)\geq 0$ for all $\zeta\in\mathscr{D}(\R)$. 
This is a direct consequence of Theorem~2.2 in~\cite{F00} and the form
of $\hat{\rho}_1$. In particular, each integrand in the bottom line
in~\eqref{eqn_rho1intident} is non-negative. Finally, the integrals
converge, because (as shown in~\cite{F00}) 
$\varphi^* (\hat{\rho}\omega_2^\Psi) (\overline{ f_\alpha},f_\alpha  )$ is of rapid
 decay as $\alpha\rightarrow +\infty$ for any Hadamard state $\Psi$ and
partial differential operator~$\hat{\rho}$ with smooth coefficients.

To treat the second term on the right-hand side
in~\eqref{first_glob_hyp_aver}, we will need the following identity, proved in appendix~\ref{app_identity}:
\begin{theorem}\label{thm_symident}
Let $F$ be a smooth function on $M\times M$ and $\partial$ be a partial differential
operator of the form $\partial=\zeta^\mu\nabla_\mu$, with a smooth vector field
$\zeta$. Then
\begin{eqnarray}\label{eqn_symident}
\lefteqn{2h^2[(\1\otimes_\mathfrak{s}\partial^2)F]_c +\partial \Big[(\1\otimes_\mathfrak{s}\partial h^2)F-(\1\otimes_\mathfrak{s}h^2\partial)F\Big]_c}\nonumber\\
&=&-2[(\partial h\otimes_\mathfrak{s}\partial h)F]_c+2(\partial h)^2[(\1\otimes_\mathfrak{s}\1)F]_c+2\partial [(h\otimes_\mathfrak{s}\partial h)F]_c,
\end{eqnarray}
for $h\in \mathscr{D}(M,\R)$.
\end{theorem}
Now choose a function $f_\mathcal{T}\in\mathscr{D}(\mathcal{T},\R)$ such that~$f_\mathcal{T}\circ \gamma =f $. 
Applying theorem~\ref{thm_symident} with~$F=\wick{\omega_2^\Psi}$, $h=f_\mathcal{T}$ and~$\partial=\zeta^\mu\nabla_\mu$, where $\zeta$ is some smooth vector field on $M$ with the property that $\zeta|_{\gamma}=\dot \gamma$, yields the identity
\begin{eqnarray}\label{eqn_symidentappl}
\lefteqn{2f_\mathcal{T}^2[(\1\otimes_\mathfrak{s}\nabla_{\zeta}^2)\wick{\omega_2^\Psi}]_c +\nabla_{\zeta} \big[(\1\otimes_\mathfrak{s}\nabla_{\zeta}f_\mathcal{T}^2)\wick{\omega_2^\Psi}-(\1\otimes_\mathfrak{s}f_\mathcal{T}^2\nabla_{\zeta})\wick{\omega_2^\Psi}\big]_c}\nonumber\\
&=&-2[(\nabla_{\zeta}f_\mathcal{T}\otimes_\mathfrak{s}\nabla_{\zeta} f_\mathcal{T})\wick{\omega_2^\Psi}]_c+2\nabla_{\zeta} f_\mathcal{T}^2[\wick{\omega_2^\Psi}]_c+2\nabla_{\zeta} \big[(f_\mathcal{T}\otimes_\mathfrak{s}\nabla_{\zeta} f_\mathcal{T})\wick{\omega_2^\Psi}\big]_c.
\end{eqnarray}
All expressions are well defined, since $\wick{\omega_2^\Psi}$ is smooth.
The first expression on the left-hand side is nothing but $f_\mathcal{T}^2[\hat{\rho}_2\wick{\omega_2^\Psi}]_c$.
Let us turn to the other terms. Recalling that $\tau $ is the proper
time along $\gamma$, we can write
\begin{eqnarray}
\int \D \tau \  \left(\nabla_{\dot{\gamma}} [(\1\otimes_\mathfrak{s}\nabla_{\dot{\gamma}}f_\mathcal{T}^2)\wick{\omega_2^\Psi}]_c\circ\gamma \right) (\tau)=\int \D \tau \  \partial_\tau\left( [(\1\otimes_\mathfrak{s}\nabla_{\dot{\gamma}}f_\mathcal{T}^2)\wick{\omega_2^\Psi}]_c\circ\gamma \right) (\tau),
\end{eqnarray} 
which vanishes, since $f_\mathcal{T}$ is of compact support.
Further terms in~\eqref{eqn_symidentappl} will vanish for the same reasons after integration. We finally obtain that
\begin{eqnarray}\label{eqn_intsymidentappl}
\int \D \tau \ f^2(\tau)\ \left([\hat{\rho}_2\wick{\omega_2^\Psi}]_c\circ \gamma\right) (\tau)&=&2\int \D \tau \ (\partial_\tau f)^2\ \varphi^*\wick{\omega_2^\Psi} (\tau,\tau)\nonumber\\
&&-2\int \D \tau \ \left([(\nabla_{\dot{\gamma}} f_\mathcal{T}\otimes_\mathfrak{s}\nabla_{\dot{\gamma}} f_\mathcal{T})\wick{\omega_2^\Psi}]_c\circ\gamma\right)(\tau ).
\end{eqnarray}
The second integral on the right-hand side may be rewritten (up to a factor), using the point-splitting trick, as 
\begin{eqnarray}\label{eqn_intsymidentappl2}
\lefteqn{\int \D \tau \ \left([(\nabla_{\dot{\gamma}} f_\mathcal{T}\otimes_\mathfrak{s}\nabla_{\dot{\gamma}} f_\mathcal{T})\wick{\omega_2^\Psi}]_c\circ\gamma\right)(\tau )}\nonumber\\
&=&\int_0^\infty \frac{\D \alpha}{\pi} \int \D \tau \D \tau' \  e^{-i\alpha (\tau -\tau ')} \partial_\tau \partial_{\tau'} \left(f(\tau) f(\tau')\ \varphi^* \wick{\omega_2^\Psi}(\tau,\tau' )\right)\nonumber\\
&=&\int_0^\infty \frac{\D \alpha}{\pi} \alpha^2 \int \D \tau \D \tau' \  \overline{ f_\alpha}(\tau)f_\alpha (\tau')\  \varphi^* \wick{\omega_2^\Psi}(\tau,\tau' )\nonumber\\
&=&\int_0^\infty \frac{\D \alpha}{\pi} \alpha^2  \ \varphi^* \omega_2^\Psi(\overline{ f_\alpha},f_\alpha  )-\int_0^\infty \frac{\D \alpha}{\pi} \alpha^2  \ \varphi^* \omega_2^0(\overline{ f_\alpha},f_\alpha  ),
\end{eqnarray}
where we have also used integration by parts in $\tau$, $\tau'$ and the fact that $f$ is of
compact  support. As before, the bottom line of~\eqref{eqn_intsymidentappl2} is a difference of
two non-negative terms.

Finally let us put the results of~\eqref{eqn_rho1intident},~\eqref{eqn_intsymidentappl}
and~\eqref{eqn_intsymidentappl2} together with the remaining term
in~\eqref{first_glob_hyp_aver}.\footnote{We do not apply the
point-splitting trick to this term as we cannot necessarily find
smooth real square roots of the geometrical quantities involved.} 
We find that
\begin{eqnarray}
\lefteqn{ \langle \rho^\textrm{quant}\circ\gamma\rangle_\Psi  (f^2)}\nonumber\\
&=&\int_0^\infty \frac{\D \alpha}{\pi}  \varphi^*(\hat{\rho}_1\  \omega_2^\Psi)\left(\overline{f_\alpha} , f_\alpha \right)-\int_0^\infty \frac{\D \alpha}{\pi}  \varphi^*(\hat{\rho}_1\  \omega_2^0)\left(\overline{f_\alpha} , f_\alpha \right)\nonumber\\
&&+2\xi \int_0^\infty \frac{\D \alpha}{\pi} \alpha^2\   \varphi^* \omega_2^\Psi(\overline{ f_\alpha},f_\alpha  )-2\xi\int_0^\infty \frac{\D \alpha}{\pi} \alpha^2  \ \varphi^* \omega_2^0(\overline{ f_\alpha},f_\alpha  )\nonumber\\
&&-2\xi\int \D \tau \ (\partial_\tau f)^2\  \varphi^*\wick{\omega_2^\Psi} (\tau,\tau) +\xi\int \D \tau \ f^2(\tau)\ \varphi^*\left( \hat{\rho}_3\wick{\omega_2^\Psi}\right)(\tau,\tau),
\end{eqnarray}
which, noting that $\varphi^*\wick{\omega_2^\Psi} (\tau,\tau)=\langle
\wick{\Phi^2}\circ\gamma\rangle_\Psi(\tau)$, results in the following
\begin{theorem}
Let $\omega_2^0$ be the two-point function of a reference Hadamard state for the
non-minimally coupled scalar field with coupling constant $\xi\in [0,1/4]$,
defined on a globally hyperbolic spacetime with smooth metric. Furthermore, let
$\gamma $ be a time-like geodesic parametrised in proper time $\tau$ and let
$f\in\mathscr{D}(\R,\R)$.  On the set of Hadamard states, we then find 
\begin{equation}\label{ineq_mink}
\left(\rho^\textrm{quant}\circ\gamma\right) (f^2) \geq -\mathfrak{Q}^{\xi}(f),
\end{equation}
where 
\begin{equation}\label{eqn_Q_decomposition}
\mathfrak{Q}^{\xi}(f)=\widetilde{\mathfrak{Q}}^\xi_A(f)\1+\xi\left(\wick{\Phi^2}\circ\gamma\right)(\mathfrak{Q}_B[f])+\xi\left(\wick{\Phi^2}\circ\gamma\right)(\mathfrak{Q}^\xi_C[f]),
\end{equation}
with 
\begin{equation}
\widetilde{\mathfrak{Q}}^{\xi}_A(f)=\int_0^\infty \frac{\D \alpha}{\pi}  \left[\varphi^*(\hat{\rho}_1\  \omega_2^0)\left(\overline{f_\alpha} , f_\alpha \right)
+2\xi \alpha^2 \  \varphi^* \omega_2^0(\overline{ f_\alpha},f_\alpha  )\right],\label{eqn_qbound}
\end{equation}
and $\mathfrak{Q}_B(f)$ and $\mathfrak{Q}^\xi_C(f)$ are functions in $\mathscr{D}(\R,\R)$ given by
\begin{eqnarray}
\mathfrak{Q}_B[f](\tau)&=&2 (\partial_\tau f(\tau))^2, \label{eqn_bbound}\\
\mathfrak{Q}^\xi_C[f](\tau)&=& f(\tau)^2\ \left(R_{\mu\nu}\gamma^\mu \gamma^\nu-\frac{1}{2}(1-4\xi)R\right)(\tau)\ \label{eqn_sbound}.
\end{eqnarray}
Furthermore $\widetilde{\mathfrak{Q}}_A^{\xi}(f)$ and $\mathfrak{Q}_B[f]$ are non-negative.
\end{theorem}
This result follows from the previous discussion by discarding
manifestly positive terms. We remark that $\widetilde{\mathfrak{Q}}^{\xi}_A(f)$
depends on the reference state and that for $\xi=0$, we recover results known for
minimal coupling~\cite{F00}. Moreover,~$\mathfrak{Q}^\xi_C(f)$ vanishes if the region of interest is Ricci-flat.

\subsection{Minkowski space}
In this section we apply the results derived in the previous subsection to
$n$-dimensional Minkowski space. Without loss of generality, we average in the time
argument $\tau=t$ at the spatial origin, i.e., $\gamma (\tau)=(\tau,\xv_0)$. 
We choose our reference state to be the vacuum state $\Omega$, which has two-point
function 
\begin{eqnarray}\label{eqn_vactwopointfctn}
\omega_2^\Omega(t,\mathbf{x},t',\mathbf{x}')&=&\int \mathrm{d}\mu (\kv) e^{-i\left[(t-t')\omega (\mathbf{k})-(\mathbf{x}-\mathbf{x}')\mathbf{k}\right]},
\end{eqnarray}
in the distributional sense. For $g\in \mathscr{D}(\R)$, we
find that
\begin{equation}\label{eqn_vacuumsmear}
\varphi^*\omega_2^\Omega\ (\overline{g}\otimes g)=\frac{1}{2}\frac{S_{n-2}}{(2\pi)^{n-1}}\int_0^\infty \mathrm{d}k\ \frac{k^{n-2}}{\omega (k)} \abs{\hat{g}\big(\omega (k)\big)}^2,
\end{equation}
where $S_{n-2}$ is the surface area of the $(n-2)$ dimensional standard unit sphere.\footnote{We have $S_m=2\sqrt{\pi^m}/\Gamma (m/2)$, where $\Gamma$ is the Gamma function.}
We also have the identity
\begin{equation}\label{eqn_mink_waveq}
m^2\varphi^*\omega_2^\Omega (\overline{g}\otimes g)+
\sum_{i=1}^{n-1}\varphi^*\left((\partial_i\otimes\partial_i)\omega_2^\Omega\right)(\overline{g}\otimes g)=\varphi^*\left((\partial_0\otimes\partial_0)\omega_2^\Omega \right)(\overline{g}\otimes g),
\end{equation}
which follows from the spacetime translation invariance of the vacuum and the field equation~\eqref{eqn_waveeqn}.
So one can absorb the mass term and the spatial derivatives
appearing in the definition of $\widetilde{\mathfrak{Q}}^\xi_A(f)$ (via
$\hat{\rho}_1$) into a further term that involves time derivatives. 
One finds that~\eqref{ineq_mink} becomes
\begin{theorem}\label{thm_minkbound}
For the non-minimally coupled scalar quantum field in $n$-dimensional Minkowski space~$\mathbf{M}^n_{\textrm{Mink}}$, 
\begin{equation}
(\rho^\textrm{quant}\circ\gamma) (f^2) \geq -\mathfrak{Q}^{\xi}(f)=
-\left(\widetilde{\mathfrak{Q}}^\xi_A(f)\1+\xi\left(\wick{\Phi^2}\circ\gamma\right)(\mathfrak{Q}_B[f])\right),
\end{equation}
in the sense of quadratic forms on Hadamard states, where
\begin{equation}\label{eqn_Qinmink}
\widetilde{\mathfrak{Q}}_A^{\xi}(f)=\frac{S_{n-2}}{(2\pi)^{n}}\int_0^\infty \D \alpha  \int_0^\infty \mathrm{d}k\ \frac{k^{n-2}}{\omega (k)}\left((1-2\xi)\omega^2(k)+2\xi \alpha^2\right)  \abs{\hat{f}(\alpha+\omega (k))}^2
\end{equation}
and
\begin{equation}
\mathfrak{Q}_B[f](t)=2\left(\partial_t f(t)\right)^2 
\end{equation}
for $f\in \mathscr{D}(\R,\R)$ and $\xi\in [0,1/4]$.
\end{theorem}
It is easy to see that $\widetilde{\mathfrak{Q}}_A^{\xi}(f)$ is non-negative for $\xi\in
[0,1/4]$ and that $\widetilde{\mathfrak{Q}}^{\xi=0}_A(f)=\mathfrak{Q}(f)$, 
where $\mathfrak{Q}(f)$ is the lower bound found in~\cite{FE98} for the minimally coupled ($\xi=0$)
scalar field. As a consequence, theorem~\ref{thm_minkbound} recovers the results of~\cite{FE98} for minimal coupling.  

We can write $\widetilde{\mathfrak{Q}}_A^{\xi}(f)$ in the form\footnote{We will assume that $n>2$ and $m>0$, but one can find similar expressions for these cases as well.}
\begin{eqnarray}
\lefteqn{\widetilde{\mathfrak{Q}}_A^{\xi}(f)
= \frac{S_{n-2}}{(2\pi)^{n}}   \int_m^\infty \D u\  |\hat{f}|^2(u) u^n} \nonumber\\
&&\times\left(\frac{1}{n}Q_{n,2}\left(\frac{u}{m}\right)-4\xi \frac{1}{n-1} Q_{n,1}\left(\frac{u}{m}\right)+2\xi \frac{1}{n-2}Q_{n,0}\left(\frac{u}{m}\right)\right) ,
\end{eqnarray}
where the non-negative functions $Q_{n,k}$ are defined by
\begin{equation}
Q_{n,k}(y)=\frac{n+k-2}{y^{n+k-2}}\int_1^y\D x  \  (x^2-1)^{(n-3)/2}x^k,
\end{equation}
for $n+k\geq 2$ .
They vanish for $k+n=2$ and else have the properties that $Q_{n,k}(1)=0$ and $Q_{n,k}(y)\to 1$ as $y\to
\infty$. As they are continuous, they are therefore bounded. 
It might be useful to use the following estimate for $\widetilde{\mathfrak{Q}}_A^{\xi}(f)$, which follows from the previous discussion for $n>2$ and $\xi\in [0,1/4]$,
\begin{equation}\label{eqn_boundforQA}
\widetilde{\mathfrak{Q}}_A^{\xi}(f)\leq \frac{S_{n-2}}{(2\pi)^{n}} \frac{3n-4}{2n(n-2)}  \int_0^\infty \D u\  |\hat{f}|^2(u) u^n.
\end{equation}
This estimate is true for the massive and the massless case and one can find a similar estimate for the two-dimensional case.

To conclude this subsection, we investigate the behaviour under rescaling of the averaging function.
The smearing function $f_\lambda (t)=f(t/\lambda)/\sqrt{\lambda}$, for $\lambda>0$,
has the property that $\|f_\lambda\|_{L^2}=\|f\|_{L^2}$, and its Fourier transform satisfies the identity
\begin{equation}
\left(\hat{f}_\lambda (u)\right)^2=\lambda \left(\hat{f}(\lambda u)\right)^2.
\end{equation}
One can conclude from this that~$\widetilde{\mathfrak{Q}}_A^{\xi}(f_\lambda)=O(\lambda^{-n})$ as $\lambda\to \infty$.
In fact, even faster decay could be concluded if $m>0$ using the arguments of~\cite{EF07}.
It follows that $\lambda\widetilde{\mathfrak{Q}}_A^{\xi}(f_\lambda)\to 0$ as $\lambda\to \infty$. For states with
$\abs{\langle\wick{\Phi^2}\circ \gamma\rangle_\Psi (t)}<c(1+|t|)^{1-\varepsilon}$, for some positive constants $c,\varepsilon$
one can show that $\lambda\langle\wick{\Phi^2}\circ \gamma\rangle_\Psi (\mathfrak{Q}_B(f_\lambda))\to 0$ as $\lambda\to \infty$, and we therefore obtain the averaged weak energy condition (AWEC) for $\xi\in [0,1/4]$ in the form
\begin{equation}\label{eq:ANEC}
\liminf_{\lambda\to\infty}\int\D t \ \langle \rho^\textrm{quant}\circ\gamma\rangle_\Psi (t) \  f (t/\lambda)^2\geq 0.
\end{equation}

This is in line with a result in~\cite{K91}, which shows that AWEC holds for states in which the particle number and the energy is
bounded (see the penultimate paragraph in section III of~\cite{K91}). In the
minimally coupled case, AWEC is also known to follow from
QEIs~\cite{FR95,FE98}. We will return to the AWEC briefly
in section~\ref{section_energetics}.

\section{Investigation of state-dependence}\label{section_nontriviality}
The lower bounds that we derived have the characteristic that they are state dependent except for minimal coupling. 
As all previously known QEIs are state independent it is important to understand the nature of the state dependence to ensure that our bounds are not vacuous.

\subsection{KMS states and temperature scaling}
In this subsection we analyse the temperature scaling behaviour of the stress energy tensor
and the bound in theorem \ref{thm_minkbound} for a KMS state $\Psi^{\beta}$,
i.e., a thermal equilibrium state at positive temperature~$\beta^{-1}$ as seen by the observer on~$\gamma$.
Its two-point function $\omega_2^\beta$ in a $n$-dimensional Minkowski space, with~$n>3$, is given by
\begin{equation}
\omega_2^\beta (t,\mathbf{x},t',\mathbf{x}')
=\int \D\mu (\kv)\ \Big( \frac{e^{-i\left((t-t')\omega (\mathbf{k})-(\mathbf{x}-\mathbf{x}')\mathbf{k}\right)}}{1-e^{-\beta \omega (\mathbf{k})}}+\frac{e^{+i\left((t-t')\omega (\mathbf{k})-(\mathbf{x}-\mathbf{x}')\mathbf{k}\right)}}{e^{\beta \omega (\mathbf{k})}-1} \Big).
\end{equation}
We renormalise the two-point function of the KMS-state by subtracting the two-point function of the vacuum~\eqref{eqn_vactwopointfctn}.
In the coincidence limit, we find that 
\begin{eqnarray}\label{eqn_KMScoincidence}
[\wick{\omega^\beta_2}]_c(t,\xv)&=&\int \D \mu (\kv ) \ \frac{2}{e^{\beta \omega (\mathbf{k})}-1} \nonumber\\
&=&B_{n,0}(\beta m),
\end{eqnarray}
with the positive function $B_{n,k}$ defined on $[0,\infty)$ for~$k\geq 0$ by
\begin{equation*}
B_{n,k}(\alpha)=\frac{S_{n-2}}{(2\pi)^{n-1}}\int_{\alpha}^\infty \mathrm{d}z\ \ (z^2-\alpha^2)^\frac{n-3}{2} \frac{z^k}{e^z-1}.
\end{equation*}
The expression~\eqref{eqn_KMScoincidence} is positive and invariant under spacetime translations.
Thus the state-dependent part of the lower bound in theorem \ref{thm_minkbound} is given by
\begin{equation}
\langle\wick{\Phi^2}\circ \gamma\rangle_{\Psi^\beta}(\mathfrak{Q}_B[f])=2\beta^{2-n}B_{n,0}(\beta m)\  \norm{f'}^2_{L^2},
\end{equation}
while the state independent part~$\widetilde{\mathfrak{Q}}^\xi (f)$, is obviously independent of $\beta$. 
On the other hand, the renormalised energy density of this state is 
\begin{eqnarray}
\langle\rho^\textrm{quant}\rangle_{\Psi^{\beta}}(t,\mathbf{x})&=&[\hat{\rho}\wick{\omega^\beta_2}]_c(t,\mathbf{x})\nonumber\\
&=&\int \mathrm{d}\mu (\kv)
\frac{2\omega^2 (\mathbf{k})}{e^{\beta \omega (\mathbf{k})}-1}\nonumber\\
&=&\beta^{-n}B_{n,2}(\beta m),
\end{eqnarray}
so that the time averaged energy density is given by
\begin{equation}
\langle\rho^\textrm{quant}\circ \gamma\rangle_{\Psi^{\beta}}( f^{2})=\beta^{-n} B_{n,2}(\beta m) \ \norm{f}^2_{L^2}.
\end{equation}
We can now state the non-triviality result:
\begin{theorem}\label{thm_nontrivial}
The bound for the energy density of a non-minimally coupled scalar quantum field given in theorem~\ref{thm_minkbound} is non-trivial in the sense of \cite{F06}, i.e, there do not exist constants~$c,c'$ such that
\begin{equation}
\abs{\langle\rho^\textrm{quant}\circ \gamma\rangle_{\Psi}( f^{2})}\leq c+c'\abs{\langle\mathfrak{Q}^{\xi}(f)\rangle_\Psi}
\end{equation}
for all Hadamard states $\Psi$ unless $f$ is identically zero.
\end{theorem}
\begin{proof}
{}From the previous discussion, we find that for a fixed non-trivial smearing function~$f$, in the limit of high temperatures 
\begin{eqnarray}
\lim_{\beta\rightarrow 0}\beta^{n}\langle\rho^\textrm{quant}\circ\gamma\rangle_{\Psi^{\beta}}( f^{2})&=& B_{n,2}(0)\norm{f}^2_{L^2}>0,\\
\textrm{ and }\lim_{\beta\rightarrow 0}\beta^{n}\widetilde{\mathfrak{Q}}_A^{\xi}(f)&=&
\lim_{\beta\rightarrow 0}\beta^{n} \langle\wick{\Phi^2}\circ\gamma\rangle_{\Psi^\beta}(\mathfrak{Q}_B[f])= 0.
\end{eqnarray}
Now if the bound in theorem \ref{thm_minkbound} is trivial, then there exists some constant $c$ such that~\eqref{eqn_nontriviality} holds. This implies that
\begin{equation}
0<B_{n,2}(0)\norm{f}^2_{L^2}<\lim_{\beta\rightarrow 0}\beta^{n}\left(c+c'\abs{\widetilde{\mathfrak{Q}}_A^{\xi}(f)}+c'\abs{\xi  \langle\wick{\Phi^2}\circ\gamma\rangle_{\Psi^\beta}(\mathfrak{Q}_B[f])}\right)=0,
\end{equation}
which is a contradiction.
\end{proof}
A more refined formulation of this result would be that our bound is {\it non-trivial with respect to
high temperature scaling for KMS states}.

\subsection{Energy behaviour}\label{section_energetics}
The previous results already show that the lower bound in theorem~\ref{thm_minkbound} has a different scaling behaviour from 
the energy density itself. We will now present a more general analysis that gives more insight into
this, again working in four-dimensional Minkowski space~$\mathbf{M}^4_{\textrm{Mink}}$. 

The lower bound in theorem~\ref{thm_minkbound} only depends on the fields $\wick{\Phi^2}$ and
$\1$. On the other hand, the energy density also involves terms such as $\wick{\dot{\Phi}^2}$.
This results in a crucial difference in their energy behaviour. 
To be more precise, we seek values of $p,q\in\R^+$ for which there are
constants $c_{f,q},c_{f,p}'$ such that
\begin{equation}\label{eqn_sandwichbound}
 c_{f,q} \ (H+m\1)^{q}\geq (\rho^\textrm{quant}\circ\gamma) (f^2)\geq  -c_{f,p}'\ (H+m\1 )^{p}
\end{equation}
holds (in the sense of quadratic forms) on the set of Hadamard vector
states. In the minimally coupled
case we know that there is a state-independent lower bound so we may
take $p=0$; however for $\xi\in(0,1/4]$ we have already shown that $(\rho^\textrm{quant}\circ\gamma)
(f^2)$ is unbounded from below, so $p$ must be strictly positive if
\eqref{eqn_sandwichbound} is to hold. By 
theorem~\ref{thm_minkbound} it is enough to show that
$\mathfrak{Q}^\xi(f)\leq c_{f,p}'(H+m\1 )^{p}$ to conclude that the
right-hand inequality in \eqref{eqn_sandwichbound} holds; we will show
that this is possible for any $p>2$. 

On the other hand, we will show that the left-hand inequality in
\eqref{eqn_sandwichbound} cannot be satisfied for $q<3$. In this sense,
our lower bound represents a non-trivial constraint. Indeed, the
situation here is reminiscent of the sharp G{\aa}rding inequalities
studied in the theory of pseudodifferential operators, in which
operators with positive classical symbols may be bounded from below
`with a gain in derivatives', i.e., by an operator of lower order.
Although the analogy is not direct, it seems worthy of further
investigation.

As a consequence of this analysis we immediately obtain another proof of
theorem~\ref{thm_nontrivial}, namely that the lower bound in
theorem~\ref{thm_minkbound} is non-trivial in the sense of~\cite{F06}
(but this time using states in the domain of a power of the Hamiltonian
rather than KMS states). 

We begin by establishing our claim relating to the left-hand inequality
in \eqref{eqn_sandwichbound}. To do this we first consider the massless
field and construct the following one-particle state 
\begin{equation}
\Psi_{\kappa}= \frac{4\pi}{\kappa}\int\D\mu (\kv) e^{-\abs{\kv} /\kappa}a^\dagger (\kv)\ \Omega.
\end{equation}
It is straightforward to calculate that $\langle H^j\rangle_{\Psi_{\kappa}}=\left(\frac{\kappa}{2}\right)^j (j+1)!$ and that\footnote{For the concrete calculation in \eqref{eqnconcretecalculationatorigin} we assumed for simplicity that the geodesic $\gamma$ is located at the spatial origin.} 
\begin{equation}\label{eqnconcretecalculationatorigin}
\langle \rho^\textrm{quant}\circ\gamma \rangle_{\Psi_{\kappa}}(t)=\frac{\kappa^4}{\pi^2}\left\{\frac{4}{1+t^2\kappa^2}-4\xi\frac{6(t^2\kappa^2-1)}{(1+t^2\kappa^2)^4}\right\}.
\end{equation}
We have a peak at $t=0$, where the expectation value scales like the fourth power in the Hamiltonian. However, this pointwise behaviour will not hold for the smeared field.  
Let us assume that $f$ is an non-trivial, integrable function. We find that 
\begin{equation} \label{eqn_asympscaling}
\lim_{\kappa\to\infty}\frac{\langle \rho^\textrm{quant}\circ\gamma\rangle_{\Psi_{\kappa}}(f^2)}{\langle H^3\rangle_{\Psi_{\kappa}} \norm{f}_{L^2}^2}=\lim_{\kappa\to\infty}\frac{2\kappa^3(2+3\xi)/\pi}{3\kappa^3}=\frac{2(2+3\xi)}{3\pi},
\end{equation}
i.e., we have an asymptotic scaling like the expectation value of~$H^3$. 
In the limit where $\kappa$ becomes large, the high momenta are dominant. 
If the field is massive, the mass would therefore become negligible. 
We can deduce that the same scaling behaviour remains true for the massive
field, at least asymptotically. 
So the smeared energy density for the scalar field in Minkowski space
scales asymptotically at least with the third power of the Hamiltonian,
i.e., the left-hand inequality in \eqref{eqn_sandwichbound} can only
hold if $q\geq 3$. It is worth remarking that we could have conducted
the same analysis using the one-particle states investigated in section
\ref{sec_negenergexample}, and with the same result: although the energy
density in these states is negative at the spacetime origin, one may
check that the smeared energy density is positive for sufficiently large $\kappa$. 

The second part of the discussion aims to establish suitable $H$-bounds
on the state-dependent part of our lower bound,
i.e., the field~$\wick{\Phi^2}$. $H$-bounds have been
discussed elsewhere (see e.g.,~\cite{FH81}), but we require more
detailed information on the power of the Hamiltonian involved and on the
controlling constants than we have been able to locate in the literature.
The following discussion may therefore be of
independent interest. 

We will now assume that $m>0$, but allow general spacetime dimensions
$n\geq 2$. Let $h$ be the one-particle Hamiltonian and $\D\Gamma$ be the second quantisation
map (so, for example the Hamiltonian is $H=\D\Gamma(h)$). 
Let us initially restrict to the domain $D_\mathscr{S}\subset\mathcal{F}_s (\mathcal{H})$, defined as the
space of vectors in Fock space all of whose $n$-particle wavefunctions are Schwartz
functions and all but finitely many of which vanish identically,
see section X.7 in~\cite{RS75}. 
This is a dense domain in the Fock space and for every vector $\Psi\in D_\mathscr{S}$, we find that
$\kv\mapsto \norm{a (\kv)\Psi}^2 \in\mathscr{S}(\R^{n-1})$ and that
\begin{equation}
\norm{\D\Gamma (h^p)^{1/2}\Psi}^2=\int\D\mu (\kv) \omega^p (\kv) \norm{a(\kv)\Psi}^2
\end{equation}
for any $p\in\R$. 

Let $g\in \mathscr{S}(\R^{n-1})\subset \mathcal{H}$. Then
$\omega^{-p/2}g\in\mathcal{H}$, where $\omega$ acts on $\mathcal{H}$ by multiplication and the Cauchy-Schwarz inequality implies that 
\begin{eqnarray}\label{eqn_annihil_bound}
\norm{a (g)\Psi}^2&\leq&\int \D\mu (\kv)\D\mu (\kv')\abs{\langle a(\kv)\Psi |a(\kv')\Psi\rangle}\abs{g (\kv)}\abs{g(\kv')}\nonumber\\
&\leq&\left\{\int \D\mu (\kv)\norm{a(\kv)\Psi}\abs{g (\kv)}\right\}^2\nonumber\\
&=&\left\{\int \D\mu (\kv)\left(\omega^{p/2}(\kv)\norm{a(\kv)\Psi}\right)\abs{\omega^{-p/2}g (\kv)}\right\}^2\nonumber\\
&\leq&\norm{\D\Gamma (h^p )^{1/2}\Psi}^2 \cdot\norm{\omega^{-p/2}g}^2_\mathcal{H}\quad p\in \R,
\end{eqnarray}
for $\Psi\in D_\mathscr{S}$. Due to the construction of $D_\mathscr{S}$, one can find that  $D_\mathscr{S}\subset  D(H^p)$ for any $p\in\R$.
Now let $\Psi^{(l)}\subset D_\mathscr{S}$ be a $l$-particle state. We have
\begin{eqnarray}
\left(H^p \Psi^{(l)}\right)(\kv_1,\dots,\kv_l)&=&\left(\sum_{i=1}^l \omega (\kv_i)\right)^p \Psi^{(l)}(\kv_1,\dots,\kv_l),\\
\left(\D\Gamma (h^p) \Psi^{(l)}\right)(\kv_1,\dots,\kv_l)&=&\sum_{i=1}^l \omega^p (\kv_i) \Psi^{(l)}(\kv_1,\dots,\kv_l).
\end{eqnarray}
For~$a,b\geq 0$ and $p\geq 1$ we find that~$(a+b)^p\geq a^p+b^p$,\footnote{Since the function~$x\to x^{p-1}$ is monotone increasing for $p\geq 1$, we have
\begin{equation}
(a+b)^p=a(a+b)^{p-1}+b(a+b)^{p-1}\geq a^p+b^p.
\end{equation}} so 
\begin{equation}
\left(\sum_{i=1}^l \omega (\kv_i)\right)^p \geq \sum_{i=1}^l \omega^p (\kv_i).
\end{equation}
This implies that $0\leq\D\Gamma (h^p) \leq H^p$ on $D_\mathscr{S}$ for $p\geq 1$. Using this, the commutation relations and~\eqref{eqn_annihil_bound}, we recover $H$-bounds of the form
\begin{eqnarray}
\norm{a (g)\Psi}^2&\leq& \norm{H^{p/2}\Psi}^2 \cdot\norm{\omega^{-p/2}g}_\mathcal{H}^2,\nonumber\\
\norm{a^\dagger (g)\Psi}^2&\leq& \norm{H^{p/2}\Psi}^2 \cdot\norm{\omega^{-p/2}g}_\mathcal{H}^2+\norm{g}_\mathcal{H}^2\norm{\Psi}^2,
\end{eqnarray}
for $ p\geq 1$ on~$D_\mathscr{S}$.

Now define the distribution space,
\begin{equation} 
\mathcal{L}_q=\{F\in  \mathscr{E}' (\mathbf{M}^n_{\textrm{Mink}}) |  \|F\|_q^2:=\|\omega^q\widetilde{F}\|_\mathcal{H}^2<\infty \},
\end{equation}
where $\norm{\cdot}_q$ is a semi-norm. Since the field is massive we have the inclusion $ \mathcal{L}_0  \subset \mathcal{L}_q$  for $q\leq0$. The field~$\Phi (F)$ defines an operator on Fock space on the domain $D\left((N+1)^{1/2}\right)$, if $F,\overline{F}\in \mathcal{L}_0$ (or equivalently~$\widetilde{F},\widetilde{\overline{F}}\in\mathcal{H}$). Since $D_\mathscr{S}\subset D\left((N+1)^{1/2}\right)$, $\Phi (F)$ is well-defined on $D_\mathscr{S}$.

Now assume $p\geq 1$, $\Psi\in D_\mathscr{S}$ and $F,\overline{F}\in \mathcal{L}_0$.
Using the inequality $(u+v)^2\leq 2(u^2+v^2)$
we find 
\begin{eqnarray}\label{eqn_bound2}
\lefteqn{\norm{\Phi (F)\Psi}^2-\norm{\Phi (F)\Omega}^2}\nonumber\\
&\leq&\left(\norm{a(\widetilde{\overline{F}})\Psi}+\norm{a^\dagger(\widetilde{F})\Psi}\right)^2\nonumber\\
&\leq& 2\norm{H^{p/2}\Psi}^2\cdot \left(\norm{\overline{F}}_{-p/2}^2+\norm{F}^2_{-p/2}\right)+\norm{F}_0^2\norm{\Psi}^2.
\end{eqnarray}
Since this inequality is valid on $D_\mathscr{S}$, which was dense in the Fock space $\mathcal{F}_s(\mathcal{H})$, this result extends to all~$\Psi\in D(H^{p/2})$.

In order to apply these results in conjunction with the point-splitting trick, we have to establish a connection between pulled back two-point functions and their representation in terms of distributionally smeared fields acting as operators on a certain domain. In particular one has to check that for some $F(x)=\zeta (t)\otimes \delta_{\xv_0}(\xv)$ with $\zeta\in\mathscr{D}(\R)$, we have
\begin{equation}\label{eqn_singularconvergence}
(\varphi^*\omega_2^\Psi)(\overline{\zeta},\zeta)=\norm{\Phi (F)\Psi}^2.
\end{equation}
The left-hand side is well-defined as discussed before. The right-hand side is well-defined as
$F,\overline{F}\in\mathcal{L}_0$. The identity can then be shown by constructing a
sequence $F_r\in \mathscr{D}(\mathbf{M}^n_{\textrm{Mink}})$, with
$\omega_2^\Psi(\overline{F_r},F_r)\to(\varphi^*\omega_2^\Psi)(\overline{\zeta},\zeta)$
and $\widetilde{F_r}\to \widetilde{F}=\widetilde{\zeta\otimes\delta_{\xv_0}}$ in 
$\mathcal{H}$.\footnote{One can do this by defining~$F_r (t,\xv)=\zeta(t)\chi_r (\xv)$, with the approximate identity~$\chi_r\in\mathscr{D}(\R^{n-1})$.}
The latter property ensures that $\Phi (F_r)\Psi\to\Phi
(\zeta\otimes\delta_{\xv_0})\Psi$ and the identity therefore holds because
$\omega_2^\Psi(\overline{F_r},F_r) =\norm{\Phi(F_r)\Psi}^2$ for test
functions $F_r\in \mathscr{D}(\mathbf{M}^n_{\textrm{Mink}})$.

We are now able to apply the above result to find $H$-bounds on the Wick square, smeared along the inertial curve~$\gamma$.
Applying the point-splitting trick and~\eqref{eqn_bound2}, we find that
\begin{eqnarray}\label{eqn_split1}
\lefteqn{\langle\wick{\Phi^2}\circ \gamma\rangle_\Psi(f^2)}\nonumber\\
&=&\int_0^\infty\frac{\D\alpha}{\pi}\left(\norm{\Phi (f_\alpha\otimes\delta_{\xv_0})\Psi}^2-\norm{\Phi (f_\alpha\otimes\delta_{\xv_0})\Omega}^2 \right)\nonumber\\
&\leq&2\norm{H^{p/2}\Psi}^2\cdot\int_0^\infty\frac{\D\alpha}{\pi} \left(\norm{\overline{f_\alpha\otimes\delta_{\xv_0}}}^2_{-p/2}+\norm{f_\alpha\otimes\delta_{\xv_0}}^2_{-p/2}\right)\nonumber\\
&&+\norm{\Psi}^2\ \int_0^\infty\frac{\D\alpha}{\pi} \norm{f_\alpha\otimes\delta_{\xv_0}}_0^2\nonumber\\
&=&2\norm{H^{p/2}\Psi}^2\cdot\int_{-\infty}^\infty\frac{\D\alpha}{\pi} \norm{f_\alpha\otimes\delta_{\xv_0}}^2_{-p/2}+\norm{\Psi}^2\ \int_0^\infty\frac{\D\alpha}{\pi} \norm{f_\alpha\otimes\delta_{\xv_0}}_0^2,
\end{eqnarray}
where $f\in\mathscr{D}(\R,\R)$ and we used~\eqref{eqn_singularconvergence} with $F(x)=f(t) e^{i\alpha t}\otimes \delta_{\xv_0}(\xv)$. 
For convenience, let us introduce the positive quadratic functionals
\begin{eqnarray}
B_{-p/2}(f)&=&2\int_{-\infty}^\infty\frac{\D\alpha}{\pi} \norm{f_\alpha\otimes\delta_{\xv_0}}^2_{-p/2}\label{defnBpf}\\
C_0(f)&=& \int_0^\infty\frac{\D\alpha}{\pi} \norm{f_\alpha\otimes\delta_{\xv_0}}_0^2.\label{defnC0f}
\end{eqnarray}
\begin{lemma}\label{lemm_normbound}
Let $(p+2)>n$, where $n$ is the spacetime dimension, and let $f\in \mathscr{S}(\R,\R)$. For the massive case,  then
$C_0(f)<\infty$ and $B_{-p/2}(f)<\infty$. 
\end{lemma}
\begin{proof}
To show that $\abs{C_0(f)}<\infty$, we see that
\begin{equation}
\abs{\widetilde{f_\alpha\otimes\delta_{\xv_0}}(\kv)}^2=\abs{\hat{f}(\omega (\kv)+\alpha)}^2\leq \frac{c}{(m+\omega(\kv)+\alpha)^{n+1}}\leq\frac{1}{(m+\alpha)^2}\frac{c}{\omega^{n-1} (\kv)},
\end{equation}
for some positive constant $c$, where we made use of the fact that $f$ is in Schwartz space.
It follows that there exists a constant $c'$, such that
\begin{equation}
\ \norm{f_\alpha \otimes\delta_{\xv_0}}_0^2\leq\frac{c'}{(m+\alpha)^2},
\end{equation}
which is integrable in $\alpha$ on $\R^+$ proving that $C_0(f)<\infty$.

To show that $\abs{B_{-p/2}(f)}<\infty$, realise that there exists a constant~$c$, such that
\begin{equation}
\abs{\widetilde{f_\alpha\otimes\delta_{\xv_0}}(\kv)}^2=\abs{\hat{f}(\omega (\kv)+\alpha)}^2\leq \frac{c}{(m^2+(\omega(\kv)+\alpha)^{2})}.
\end{equation}
Using this inequality, we get
\begin{eqnarray}
\int_{-\infty}^\infty\frac{\D\alpha}{\pi} \norm{f_\alpha\otimes\delta_{\xv_0}}^2_{-p/2}&\leq &\int \D \mu (\kv)\ \omega^{-p}(\kv)  \  \int_{-\infty}^\infty\frac{\D\alpha}{\pi} \frac{c}{(m^2+(\omega(\kv)+\alpha)^{2})}\nonumber\\
&<&\frac{c}{m} \int \D \mu (\kv) \ \omega^{-p} (\kv),
\end{eqnarray}
where we made use of Tonelli's theorem.
The last expression, however, is finite due to the restriction on $p$.
\end{proof}
We note that, under the assumptions of lemma~\ref{lemm_normbound}, 
the quantities in \eqref{defnBpf} and \eqref{defnC0f} may be estimated 
by similar arguments to those used to obtain~\eqref{eqn_boundforQA}.
Combining inequality~\eqref{eqn_split1} and lemma~\ref{lemm_normbound},
we can now summarise our result as the following $H$-bound: 
\begin{theorem}\label{thm_wickbound}
Let $n$ be the spacetime dimension, $p>(n-2)$, $f\in\mathscr{D}(\R,\R)$ and $\Psi \in D(H^{p/2})$. Then
\begin{equation}
\langle\wick{\Phi^2}\circ \gamma\rangle_\Psi(f^2)\leq B_{-p/2}(f)\norm{H^{p/2}\Psi}^2+C_0(f)\norm{\Psi}^2 <\infty.
\end{equation}
\end{theorem}
As an immediate consequence of theorem~\ref{thm_wickbound} we have:
\begin{corollary}
Subject to the assumptions and notation of theorem~\ref{thm_minkbound} and
theorem~\ref{thm_wickbound}, we have
\begin{equation}\label{eqn_boundincoroll}
\abs{\langle\mathfrak{Q}^{\xi}\rangle_\Psi(f)}\leq
\left(\widetilde{\mathfrak{Q}}^\xi_A(f)+2\xi C_0(\partial_t f)\right)\norm{\Psi}^2+2\xi B_{-p/2}(\partial_t f) \norm{H^{p/2}\Psi}^2.
\end{equation}
It follows that
there exists a constant $c'_{f,p}$ for which
\begin{equation}\label{unter_sandwich}
(\rho^{quant}\circ\gamma) (f^2)\geq \mathfrak{Q}^{\xi}(f)\geq  -c_{f,p}'\ (H+m\1 )^{p}
\end{equation}
as an inequality of quadratic forms for Hadamard vector states $\Psi$.
\end{corollary}
In particular, we see that large negative time-averaged energy densities
(for a given smearing function) can only be obtained at large positive
overall energies. 

As an application of these results we may again consider the rescaled
test function $f_\lambda$ as in our earlier discussion of the AWEC. One
can show that $\lambda C_0(\partial_t f_\lambda)$ and  $\lambda B_{-p/2}(\partial_t
f_\lambda)$ both tend to zero as $\lambda\to\infty$ (for the allowed
values of $p$) which, together with our earlier observation that
$\lambda \widetilde{\mathfrak{Q}}^\xi_A(f_\lambda)\to 0$, yields a proof
of AWEC in the form~\eqref{eq:ANEC} for all Hadamard vector states
(which necessarily belong to the domain of $H^{p/2}$). 

Now let us return to the four-dimensional case. Our result shows that a bound of the type \eqref{unter_sandwich}
can be satisfied for any $p>2$ (although we
cannot exclude the possibility that it might also hold for some $p\leq
2$). The important point is that our QEI is a non-trivial
restriction on the averaged energy density because the \emph{upper} bound in
\eqref{eqn_sandwichbound} cannot be satisfied for any $q<3$ (although it
can be satisfied for any $q>4$ by adapting our $H$-bound arguments).

\section{Conclusion}\label{conclusion}
Our main results may be summarised as follows. First, we have shown (at least for
massless fields in four-dimensional Minkowski space)
that the non-minimally coupled scalar field with $\xi>0$ admits states with arbitrarily
large negative energy density over arbitrarily large bounded spacetime regions.
Thus the non-minimally coupled field cannot obey QEIs of the type
previously studied in the literature, in which there is a state-independent
lower bound. Second, by combining the QEI derivation for the minimally
coupled field~\cite{F00} with techniques previously applied to the
classical non-minimally coupled field~\cite{FO06}, we have derived a new
type of QEI with a state-dependent lower bound, for couplings in the
range $(0,1/4]$ in general globally hyperbolic spacetimes. Third, we
have analysed the state-dependence of the bound, which involves the Wick
square of the field, rather than the Wick powers of its derivatives that
appear in the energy density. This involved the formulation of various $H$-bounds, which may be
of independent interest. Roughly
speaking, the results of this analysis tell us that the lower bound
scales more softly with the overall energy scale than the energy density
itself. Thus we may conclude that negative energy effects with large
magnitude, while possible over large regions, require more energy to
achieve than positive energy densities of the same magnitude and that
the `energy budget' for these two effects will grow with a different
power. 

In the light of our results, it seems reasonable to expect that
generic interacting quantum fields will not obey state-independent QEIs,
as has also been argued on physical grounds for a particular model in~\cite{OG03}. The
state-independent QEIs that have been found for other free fields and
conformal fields in two dimensions (see \cite{FH05}) should therefore be
regarded as particularly simple cases. In general, then, the aim should
be to establish non-trivial state-dependent bounds. Suitable
generalisations of our $H$-bounds to higher order Wick polynomials may be useful in this context.

It also becomes important to understand whether one can draw useful physical conclusions
from state-dependent QEIs. For example, the state-independent QEIs have been used to
place constraints on exotic
spacetimes~\cite{PF97,FR96,FR05,ER97} and
are linked to questions of thermodynamic
stability~\cite{F78,FV03}. 
A first indication that state-dependent bounds can be used for similar
purposes may be found in our derivation of AWEC. In general we expect
that similar results will obtain for state-dependent bounds, once one
restricts to states of a given energy scale. Likewise, we would expect
the phenomenon of quantum interest~\cite{FR99,FT00} to be governed
by the energy scale in this context. Finally, our general QEI is of the
so-called `difference' type: it constrains the normal ordered energy
density rather than the (Hadamard) renormalised version. This
restriction has recently been removed for the minimally coupled field
in~\cite{FS07} to obtain an `absolute' QEI; one would expect that this can also be adapted to the
non-minimally coupled field.

\appendix
\section{Proof of theorem~\ref{thm_symident}}\label{app_identity}

Throughout the appendix $[\cdot]_c$ denotes the `coincidence limit', i.e., the diagonal of smooth functions on $M\times M$, but for the sake of clearity we will omit the arguments.
\begin{proof}[Proof of theorem \ref{thm_symident}]
First, we look at the top line in~\eqref{eqn_symident}. To make the calculations clearer, we start by calculating the expression without the symmetric product. 
We know by Synge's theorem that 
\begin{equation}
\partial [F]_c=[(\partial\otimes \1)F]_c+[(\1\otimes \partial)F]_c,
\end{equation}
where $F$ is a smooth function on $M\times M$, so 
\begin{eqnarray}\label{eqn_symident1}
\lefteqn{\partial [(\1\otimes\partial h^2)F-(\1\otimes h^2\partial )F]_c}\nonumber\\
&=&[(\partial\otimes\partial h^2)F+(\1\otimes\partial^2h^2)F]_c-[(\partial\otimes h^2\partial )F+(\1\otimes \partial h^2 \partial)F]_c\nonumber\\
&=&[(\partial\otimes (\partial h^2))F+(\partial\otimes h^2\partial)F]_c\nonumber\\
&&+[(\1\otimes (\partial^2h^2))F+2(\1\otimes (\partial h^2)\partial )F+(\1\otimes h^2\partial^2)F]_c\nonumber\\
&&-[(\partial\otimes h^2\partial)F]_c-[(\1\otimes (\partial h^2)\partial)F+(\1\otimes h^2\partial^2)F]_c\nonumber\\
&=&2(\partial h^2)[(\1 \otimes_\mathfrak{s}\partial)F]_c+(\partial^2h^2)[(\1\otimes\1)F]_c.
\end{eqnarray}
Where we want to emphasise that there is a symmetric product in the last line.
Now it is quite obvious that the same calculation is valid if one exchanges the arguments of $F$.
Thus adding the term $2h^2[(\1\otimes_\mathfrak{s}\partial^2)F]_c$ on both sides, we find the following identity for the symmetrised expression
\begin{eqnarray}\label{eqn_symident2}
\lefteqn{2h^2[(\1\otimes_\mathfrak{s}\partial^2)F]_c +\partial [(\1\otimes_\mathfrak{s}\partial h^2)F-(\1\otimes_\mathfrak{s} h^2\partial )F]_c}\nonumber\\
&=&2h^2[(\1\otimes_\mathfrak{s}\partial^2)F]_c+2(\partial h^2)[(\1 \otimes_\mathfrak{s}\partial)F]_c+(\partial^2h^2)[(\1\otimes_\mathfrak{s}\1)F]_c.
\end{eqnarray}
Here the top line is obviously identical with the expression in the top line of~\eqref{eqn_symident}.

Now let us do the analogous calculations for the bottom line in~\eqref{eqn_symident}. As before, we start for clarity without the symmetric product
\begin{eqnarray}\label{eqn_symident3}
\lefteqn{2\partial [(h\otimes \partial h)F]_c-2[(\partial h\otimes \partial h)F]_c+2(\partial h)^2 [(\1\otimes\1)F]_c}\nonumber\\
&=&2[(\partial h\otimes \partial h)F+(h\otimes \partial^2 h)F]_c-2[(\partial h\otimes \partial h)F]_c+2(\partial h)^2 [(\1\otimes\1)F]_c\nonumber\\
&=&2h[(\1\otimes (\partial^2h))F+2(\1\otimes (\partial h )\partial)F+(\1\otimes h \partial^2)F]_c+2(\partial h)^2 [(\1\otimes\1)F]_c\nonumber\\
&=&2h^2 [(\1\otimes \partial^2)F]_c+2(2h\partial h)[(\1\otimes\partial)F]_c+(2h\partial^2h+2(\partial h)^2)[(\1\otimes\1)F]_c\nonumber\\
&=&2h^2 [(\1\otimes \partial^2)F]_c+2(\partial h^2)[(\1\otimes\partial)F]_c+(\partial^2h^2)[(\1\otimes\1)F]_c.
\end{eqnarray}
Again, if we symmetrise this, we get
\begin{eqnarray}\label{eqn_symident4}
\lefteqn{2\partial [(h\otimes_\mathfrak{s} \partial h)F]_c-2[(\partial h\otimes_\mathfrak{s} \partial h)F]_c+2(\partial h)^2 [(\1\otimes_\mathfrak{s}\1)F]_c}\nonumber\\
&=&2h^2 [(\1\otimes_\mathfrak{s} \partial^2)F]_c+2(\partial h^2)[(\1\otimes_\mathfrak{s}\partial)F]_c+(\partial^2h^2)[(\1\otimes_\mathfrak{s}\1)F]_c.
\end{eqnarray}
Comparing the bottom row of~\eqref{eqn_symident4} with the bottom row of~\eqref{eqn_symident2} proves the theorem.
\end{proof}


\end{document}